\documentclass[a4paper,UKenglish,cleveref, autoref, thm-restate]{lipics-v2021}

\title{Faster Algorithms for Bounded Tree Edit Distance}

\author{Shyan Akmal}{MIT EECS and CSAIL, USA }{naysh@mit.edu}{https://orcid.org/0000-0002-7266-2041}{Supported by NSF Grant CCF-1909429.}

\author{Ce Jin}{MIT EECS and CSAIL, USA}{cejin@mit.edu}{}{Supported by an MIT Akamai Presidential Fellowship}

\authorrunning{S.\ Akmal and C.\ Jin} 

\Copyright{Shyan Akmal and Ce Jin} 



\category{Track A: Algorithms, Complexity and Games} 

\relatedversion{} 

\acknowledgements{We thank Virginia Vassilevska Williams
for several helpful discussions.}

\nolinenumbers 

\hideLIPIcs  

\EventEditors{Nikhil Bansal, Emanuela Merelli, and James Worrell}
\EventNoEds{3}
\EventLongTitle{48th International Colloquium on Automata, Languages, and Programming (ICALP 2021)}
\EventShortTitle{ICALP 2021}
\EventAcronym{ICALP}
\EventYear{2021}
\EventDate{July 12--16, 2021}
\EventLocation{Glasgow, Scotland (Virtual Conference)}
\EventLogo{}
\SeriesVolume{198}
\ArticleNo{12}

\renewcommand{\epsilon}{\varepsilon}
\newcommand{\eps}{\varepsilon}

\renewcommand{\tilde}{\widetilde}

\newcommand{\ed}{\mathsf{ed}}
\newcommand{\size}{\mathsf{size}}
\newcommand{\lca}{\mathsf{LCA}}
\newcommand{\pre}{\mathsf{pre}}
\newcommand{\parent}{\mathsf{par}}

\newcommand{\grp}[1]{\left(#1\right)}

\usepackage{caption}
\usepackage{transparent}

\begin{document}

\maketitle

\begin{abstract}
\emph{Tree edit distance} is a well-studied measure of dissimilarity between rooted trees with node labels.
It can be computed in $O(n^3)$ time [Demaine, Mozes, Rossman, and Weimann, ICALP 2007], and fine-grained hardness results suggest that the weighted version of this problem cannot be solved in truly subcubic time unless the APSP conjecture is false [Bringmann, Gawrychowski, Mozes, and Weimann, SODA 2018].

We consider the \emph{unweighted} version of tree edit distance, where every insertion, deletion, or relabeling operation has unit cost. 
Given a parameter $k$ as an upper bound on the distance, the previous fastest algorithm for this problem runs in $O(nk^3)$ time [Touzet, CPM 2005], which improves upon the cubic-time algorithm for $k\ll n^{2/3}$. 
In this paper, we give a faster algorithm taking $O(nk^2 \log n)$ time, improving both of the previous results for almost the full range of $\log n \ll k\ll n/\sqrt{\log n}$. 
\end{abstract}

\section{Introduction}
\label{sec:intro}
Many tasks involve measuring the similarity between two sets of data.
When the data is naturally represented as a string of characters, one of the most popular and well-studied ways of measuring similarity is via the (string) edit distance, defined to be the minimum number of characters that must be deleted, inserted, and substituted to turn one string into the other.
Although edit distance is a fundamental problem in computer science and has been employed to great effect in many other areas, it can be less useful
for applications where we are interested in comparing data that is not just linearly ordered, but has some hierarchical organization.
When the data admits a tree structure, a natural measure of similarity is the \emph{tree edit distance}, first introduced by Tai \cite{Tai79} as a generalization of the string edit distance problem \cite{WagnerF74}.
Computing this metric has a wide variety of applications in a diverse array of  fields including computational biology \cite{gusfield_1997,10.1093/bioinformatics/6.4.309,HochsmannTGK03,waterman1995introduction}, structured data analysis \cite{KochBG03,Chawathe99,FerraginaLMM09}, and image processing \cite{BellandoK99,KleinTSK00,KleinSK01,SebastianKK04}.

Given two \emph{rooted ordered} trees with node labels, the tree edit distance is the minimum number of node deletions, insertions, and relabelings needed to turn one tree into the other.
When we delete a node, its children become children of the parent of the deleted node. 
Beyond this widely studied definition, there are many other variants of the tree edit distance problem, including those defined for unrooted trees or unordered trees, or parameterized by the depth or the number of leaves, which we do not consider in this paper. We refer interested readers to the survey by Bille \cite{Bille05survey} for a comprehensive review.

We now recount the development of exact algorithms for tree edit distance.
In 1979, Tai \cite{Tai79} gave the first algorithm that computes the tree edit distance between two node-labeled rooted trees on $n$ nodes in $O(n^6)$ time. 
The time complexity was improved to $O(n^4)$ by Zhang and Shasha \cite{ZhangS89} using a dynamic programming approach. 
Later, Klein \cite{Klein98} applied the heavy-light decomposition technique to obtain an $O(n^3 \log n)$ time algorithm. 
Finally, Demaine, Mozes, Rossman, and Weimann \cite{demaine2009} improved the running time by a log-factor to $O(n^3)$,
and further showed that this running time is optimal among a certain class of dynamic programming algorithms termed \emph{decomposition strategy algorithms} by Dulucq and Touzet \cite{DulucqT03,DulucqT05}. 
When the two input trees have different sizes $m\le n$, their algorithm runs in $O\left (nm^2(1+\log \frac{n}{m})\right )$ time. 

All algorithms mentioned above actually compute tree edit distance in the general \emph{weighted} setting where the cost of deleting, inserting, or relabeling is a function of the labels (so that deleting nodes with certain labels might be cheaper than deleting other nodes with different labels). 
In this setting,
Bringmann, Gawrychowski, Mozes, and Weimann \cite{apsphard2020} showed conditional hardness results for the tree edit distance problem: a truly subcubic time algorithm for this problem would imply a truly subcubic time algorithm for the All-Pairs Shortest Paths (APSP) problem (assuming alphabet of size $\Theta(n)$), and an $O(n^{k(1-\eps)})$ time algorithm for the Max-weight $k$-clique problem (assuming a sufficiently large constant-size alphabet). 
However, the instances produced by their fine-grained reduction have non-unit edit costs, and 
it is not clear yet how to prove a conditional hardness result for the \emph{unweighted} tree edit distance problem with unit edit costs. 
In contrast, the quadratic-time fine-grained lower bound for the \emph{string} edit distance problem (based on the Strong Exponential Time Hypothesis) holds for unit-cost operations \cite{BackursI18,AbboudBW15}.

Therefore, it is natural to consider the unweighted unit-cost setting, where every elementary operation has cost 1, independent of the labels. 
In this case, the distance between two trees of sizes $n$ and $m$ cannot be larger than $n+m$, and is arguably even smaller in practical scenarios. 
In 2005, Touzet \cite{Touzet05,Touzet07} gave an algorithm in this context that computes the unweighted tree edit distance in $O(nk^3)$ time, assuming the distance is at most $k$. 
When $k= \Theta(n)$, Touzet's algorithm has the same performance as the $O(n^4)$ time algorithm by Zhang and Shasha \cite{ZhangS89}.
However, the running time significantly improves if the upper bound $k$ is much smaller than $n$. (There are also algorithms that run faster when the input trees have low depth or few leaf nodes, e.g. \cite{DBLP:conf/cikm/0001A20})
We remark that similar progress was shown earlier for the string edit distance problem: although the best known running time for the general case is $O(n^2/\log^2 n)$ \cite{MasekP80,tcs/BilleF08}, when the distance is at most $k$, Ukkonen \cite{Ukkonen85} gave an $O(nk)$ time algorithm, which was later improved to $\tilde O(n+k^2)$ time\footnote{In this paper, $\tilde O(f)$ stands for $f\cdot (\log f)^{O(1)}$.} by Myers \cite{Meyers86}, Landau and Vishkin \cite{LandauV88} using suffix trees.

Although we focus on exact algorithms in this work, 
approximation algorithms for the tree edit distance problem have also been studied \cite{AkutsuFT10,apxfocs2019}. Boroujeni, Ghodsi, Hajiaghayi, and Seddighin \cite{apxfocs2019} showed an algorithm that computes a $(1+\eps)$-approximation of the tree edit distance in $\tilde O(\eps^{-3} n^2)$ time.
If an upper bound $k$ on the distance is known, the running time can be improved to $\tilde O(\eps^{-3} nk)$.
For the easier problem of approximating string edit distance, there is a longer line of research \cite{AndoniO12,AndoniKO10,BoroujeniEGHS18,ChakrabortyDGKS18,BrakensiekR20,KouckyS20} culminating in a near-linear time constant-factor approximation algorithm \cite{AndoniN20}.


\subsection{Our contribution}
We present a faster algorithm for exactly computing the unweighted tree edit distance (where every elementary operation has unit cost), with a parameter $k \le O(n)$ given as an upper bound on the distance.
\begin{theorem}
\label{thm:bounded-ted}
Given two node-labeled rooted trees $T_1, T_2$ each of size at most $n$, we can compute the unweighted tree edit distance between $T_1$ and $T_2$ exactly in $O(n k^2 \log n)$ time, assuming the distance is at most $k$.
\end{theorem}

When the distance parameter $k$ is constant our algorithm runs in quasilinear time, and as $k$ reaches its upper bound $O(n)$ we recover the $O(n^3\log n)$ time algorithm by Klein \cite{Klein98}.
Our algorithm outperforms the $O(n^3)$ time algorithms of Demaine et al. \cite{demaine2009} when $k = o(n/\sqrt{\log n})$.
As mentioned earlier, the previous best algorithm for bounded tree edit distance by Touzet \cite{Touzet07} takes $O(nk^3)$ time. 
The time complexity of our algorithm improves upon this prior work whenever $k = \omega(\log n)$.

\subsection{High-level Overview}

Touzet's $O(nk^3)$-time algorithm is based on Zhang and Shasha's $O(n^4)$-time dynamic programming algorithm \cite{ZhangS89}. The improvement was achieved by pruning unuseful DP states, and only considering $O(nk^3)$ many states instead of $O(n^4)$. 
This pruning technique was inspired by an idea used in the previous $O(nk)$-time algorithm for string edit distance \cite{Ukkonen85}: for input strings whose edit distance is at most $k$, when building the dynamic programming table for computing the edit distance, it suffices to only compute entries of the table corresponding to prefixes whose lengths differ by at most $k$.
Touzet's improvement for tree edit distance employs a similar technique and relies on measuring the ``distance'' between two DP states with respect to the preorder tree traversal, which is compatible with the DP transitions of Zhang and Shasha.

We modify Klein's $O(n^3 \log n)$ time algorithm by further reducing the number of useful states, similar in spirit to the algorithm by Touzet \cite{Touzet07}. 
The main difficulty in adapting this idea is that unlike the algorithm of Zhang and Sasha, Klein's DP algorithm does not follow the same preorder traversal of the nodes.
Hence we need completely new arguments to bound the number of useful DP states.
Beyond considering the sizes of the subproblems generated,
our proofs examine how various subforests are generated by different transition rules and employ some combinatorial arguments about how the subgraphs of deleted nodes can be structured when the edit distance is known to be bounded.

\subsection{Paper Organization}
In \Cref{sec:prelim} we formally define the tree edit distance problem and introduce the notation used throughout the rest of the paper. 
Next, in \Cref{sec:klein}, we review Klein's algorithm \cite{Klein98} which our algorithm builds off of. 
Then, in \Cref{sec:improve}, we present our improved algorithm.
Finally, we conclude by mentioning several open questions relevant to our work in \Cref{sec:open}.

\section{Preliminaries}
\label{sec:prelim}

%
In this paper, we consider rooted trees that are \emph{ordered}, meaning that the order between siblings is significant. 
We also consider \emph{forests} consisting of disjoint rooted trees, where the order between these trees is also significant. It is convenient to treat the tree roots of a forest as the children of a virtual root node.
Let $\parent(v)$ denote the parent node of $v$, or the virtual root node if $v$ is a tree root in the forest.

We define the \emph{node removal operation} in the following natural way: after removing a node $v$ from the forest $F$, the children of $v$ become children of $\parent(v)$, preserving the same relative order.
We use $F-v$ to denote the forest obtained by removing $v$ from $F$. 

We now formally define the \emph{tree edit distance} as a metric on ordered rooted trees with node labels.
\begin{definition}[(Unweighted) Tree Edit Distance]
\label{defn:unrooted-ted}
Let $T_1$ and $T_2$ be two ordered rooted trees whose nodes are labeled with symbols from some alphabet $\Sigma$. 
There are two types of allowed operations:
\begin{itemize}
    \item Relabeling: change the label of a node from one symbol in $\Sigma$ to another.
    \item Deletion: remove a node. 
\end{itemize}
Then the tree edit distance between $T_1$ and $T_2$, denoted by $\ed(T_1,T_2)$, is the minimum number of operations that must be performed on $T_1$ and $T_2$ to obtain two identical forests.
\end{definition}

\Cref{fig:example-ted} provides an example of these operations in action.

\begin{figure}[t]
\centering
\def\svgwidth{\linewidth}
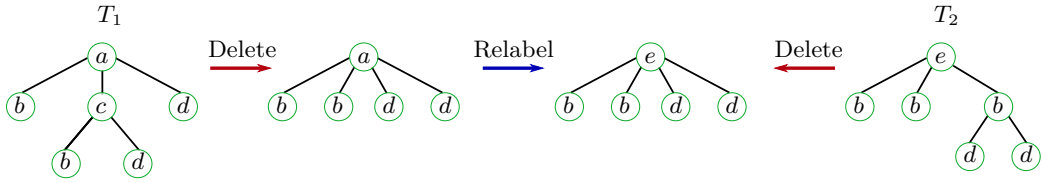
\caption{To turn $T_1$ and $T_2$ into the same tree with a minimum number of operations, we can delete a node from each and relabel a node in $T_1$. 
So in this example $\ed(T_1, T_2) = 3$.}
\label{fig:example-ted}
\end{figure}

\begin{remark}
An alternative definition of tree edit distance is the minimum number of insertions, deletions, and relabeling needed to turn one tree into the other. It is easy to see that these two definitions are equivalent.
\end{remark}

Since the operations of relabeling and deletion also apply to labeled forests, the above definition naturally extends to measure the edit distance between two \emph{forests} $F_1$ and $F_2$, and for the rest of the paper we write $\ed(F_1,F_2)$ to denote this edit distance as well.

\begin{figure}[t]
\centering
\def\svgwidth{0.6\linewidth}
\begingroup%
  \makeatletter%
  \providecommand\color[2][]{%
    \errmessage{(Inkscape) Color is used for the text in Inkscape, but the package 'color.sty' is not loaded}%
    \renewcommand\color[2][]{}%
  }%
  \providecommand\transparent[1]{%
    \errmessage{(Inkscape) Transparency is used (non-zero) for the text in Inkscape, but the package 'transparent.sty' is not loaded}%
    \renewcommand\transparent[1]{}%
  }%
  \providecommand\rotatebox[2]{#2}%
  \newcommand*\fsize{\dimexpr\f@size pt\relax}%
  \newcommand*\lineheight[1]{\fontsize{\fsize}{#1\fsize}\selectfont}%
  \ifx\svgwidth\undefined%
    \setlength{\unitlength}{243.77952756bp}%
    \ifx\svgscale\undefined%
      \relax%
    \else%
      \setlength{\unitlength}{\unitlength * \real{\svgscale}}%
    \fi%
  \else%
    \setlength{\unitlength}{\svgwidth}%
  \fi%
  \global\let\svgwidth\undefined%
  \global\let\svgscale\undefined%
  \makeatother%
  \begin{picture}(1,0.60465116)%
    \lineheight{1}%
    \setlength\tabcolsep{0pt}%
    \put(0,0){\includegraphics[width=\unitlength,page=1]{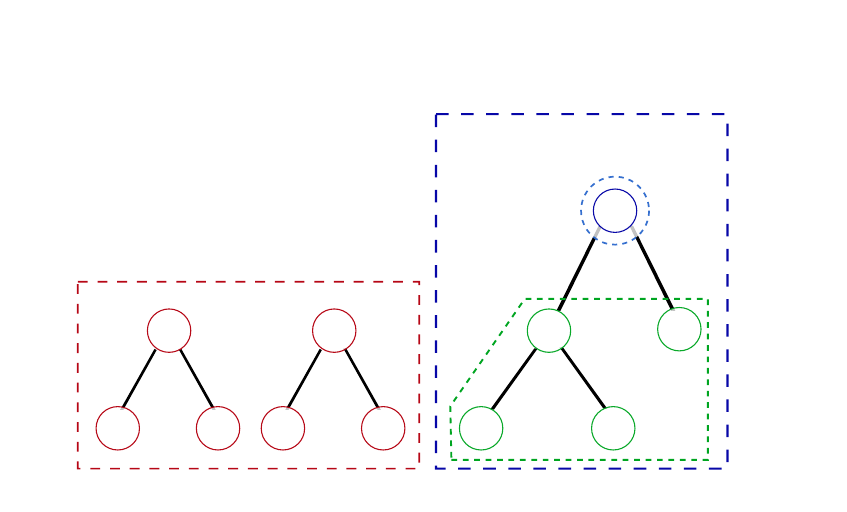}}%
    \put(0.25882532,0.31191577){\color[rgb]{0,0,0}\makebox(0,0)[lt]{\lineheight{1.25}\smash{\begin{tabular}[t]{l}$L'_F$\end{tabular}}}}%
    \put(0.65877724,0.50881501){\color[rgb]{0,0,0}\makebox(0,0)[lt]{\lineheight{1.25}\smash{\begin{tabular}[t]{l}$R_F$\end{tabular}}}}%
    \put(0.54802132,0.26884401){\color[rgb]{0,0,0}\makebox(0,0)[lt]{\lineheight{1.25}\smash{\begin{tabular}[t]{l}$R_F^\circ$\end{tabular}}}}%
    \put(0.68954273,0.41651857){\color[rgb]{0,0,0}\makebox(0,0)[lt]{\lineheight{1.25}\smash{\begin{tabular}[t]{l}$r_F$\end{tabular}}}}%
  \end{picture}%
\endgroup%

\caption{The example forest $F$ above is partitioned into $L'_F, r_F$, and $R_F^\circ$.}
\label{fig:example-lr}
\end{figure}
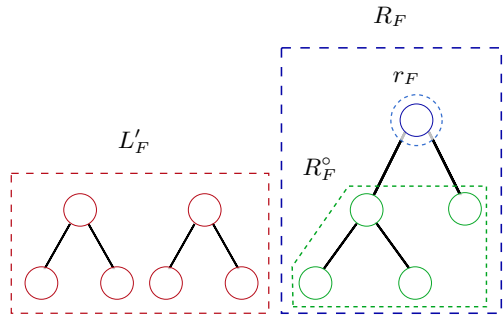

Given a forest $F$, we write $L_F$ (or $R_F$) to denote the leftmost (or rightmost) tree in $F$, and write $\ell_F$ (or $r_F$) to denote the root of $L_F$ (or $R_F$).
For convenience, let $L'_F$ denote $F-R_F$, and let $R_F^{\circ}$ denote $R_F-r_F$ (similarly, $R'_F=F-L_F$ and $L_F^{\circ}=L_F-\ell_F$). Hence, the nodes of a nonempty forest $F$ can be partitioned into three parts: $L'_F$, $r_F$, and $R_F^{\circ}$ (an example is given in \Cref{fig:example-lr}).
Finally, $\size(F)$ or $|F|$ denote the number of nodes in $F$ (where $F$ can also be any subset of nodes).

\begin{definition}[Subforest]
Given a rooted tree $T$, we say $F$ is a \emph{subforest} of $T$ if we can obtain $F$ from $T$ by repeatedly deleting the leftmost or rightmost root. 
\end{definition}
\begin{figure}[t]
\centering
\def\svgwidth{80pt}
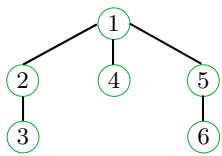
\caption{The subforests of this rooted tree are: $\{1,2,3,4,5,6\}, \{2,3,4,5,6\},\{3,4,5,6\},\{4,5,6\}$, $\{5,6\},\{6\},\emptyset$,$\{2,3,4,6\},\{3,4,6\},\{4,6\},\{2,3,4\},\{3,4\},\{4\},\{2,3\},\{3\}$. For example, the subforest $\{3,4,6\}$ can be obtained by first removing the leftmost root $1$, then removing the rightmost root $5$, and finally removing the leftmost root $2$.}
\label{fig:example-subforest}
\end{figure}
An example illustrating the definition of subforests is given in \Cref{fig:example-subforest}.
\begin{proposition}
\label{prop:n2}
A rooted tree $T$ of $n$ nodes has at most $O(n^2)$ subforests.
\end{proposition}
\begin{proof}
Although a subforest may result from interleaving operations of removing the leftmost root and removing the rightmost root, it is not hard to see that every such subforest $F$ can also be obtained from $T$ by first removing the leftmost root $a$ times, and then removing the rightmost root $b$ times, for some nonnegative integer $a,b$ with $a+b\le n$. Specifically, let $u$ be the node in $F$ with the smallest index $\pre(u)$ in the \emph{preorder traversal} of $T$ ($1\le \pre(u)\le n$), and we can set $a=\pre(u)-1$ and $b = n-a-\size(F)$.  The claim then follows from the number of choices of $(a,b)$.
\end{proof}

For a subforest $F$ of $T$, define $\lca_T(F)$ as the lowest common ancestor in $T$ of all nodes in $F$. 
When the identity $T$ is clear from context, we may write $\lca(F)$ and leave the underlying tree implicit.
Observe that $\lca(F)$ is in $F$ precisely when $F$ is a subtree of $T$.

Throughout, we use $T_1,T_2$ to denote the input trees (or $T$ if we do not specify which one of the two) we want to compute the edit distance between.

\section{Review of Klein's Algorithm}
\label{sec:klein}

We briefly review Klein's algorithm \cite{Klein98} in the context of computing the unweighted tree edit distance $\ed(T_1,T_2)$ (see \cite{demaine2009,Bille05survey,unrooted2018} for other overviews of this algorithm). 

The algorithm uses dynamic programming (DP) over pairs $(F_1,F_2)$, where $F_1, F_2$ are subforests of $T_1,T_2$, respectively. 
 Let the node relabeling cost $\delta(x,y) = 1$ if nodes $x,y$ have different labels, and $\delta(x,y) = 0$ otherwise. 
Then $\ed(F_1,F_2)$ can be computed recursively as follows \cite{ZhangS89}:
    
    \begin{itemize}
        \item 
   The base case is where either of $F_1,F_2$ is empty (denoted as $\emptyset$), and we have
    \begin{equation}
    \label{eq:trivial}
    \ed(F_1, \emptyset) = \size(F_1), \ed(\emptyset, F_2) = \size(F_2).
    \end{equation}
    
    \item 
    When both $F_1,F_2$ are nonempty, if $\size(L_{F_1}) > \size(R_{F_1})$, then we recurse with
    \begin{equation}
    \label{eq:right-recurse}
    \ed(F_1, F_2)  = \min
    \begin{cases}
        \ed(F_1-r_{F_1},F_2)+1 \\
        \ed(F_1, F_2-r_{F_2})+1 \\
        \ed(R^\circ_{F_1},R^\circ_{F_2}) + \ed(L'_{F_1}, L'_{F_2}) + \delta(r_{F_1},r_{F_2}).
    \end{cases}
    \end{equation}
    \item 
    Otherwise, $\size(L_{F_1}) \le \size(R_{F_1})$, and we recurse with
        \begin{equation}
    \label{eq:left-recurse}
    \ed(F_1, F_2)  = \min
    \begin{cases}
        \ed(F_1-\ell_{F_1},F_2)+1 \\
        \ed(F_1, F_2-\ell_{F_2})+1 \\
        \ed(L^\circ_{F_1},L^\circ_{F_2}) + \ed(R'_{F_1}, R'_{F_2}) + \delta(\ell_{F_1},\ell_{F_2}).
    \end{cases}
    \end{equation}
    \end{itemize}
    
Taking \Cref{eq:right-recurse} as an example, the recursion considers three options concerning the rightmost roots of $F_1,F_2$: (1) $r_{F_1}$ is removed. (2) $r_{F_2}$ is removed. (3) The two roots are matched to each other, generating two subproblems of matching their subtrees $R_{F_1}^\circ, R_{F_2}^\circ$, and matching the remaining parts $L'_{F_1},L'_{F_2}$. The other recursion rule in \Cref{eq:left-recurse} is symmetric and considers the leftmost roots.

We can easily verify that, if we compute $\ed(T_1,T_2)$ using this recursion, the DP states visited by the recursion are indeed pairs of subforests of $T_1$ and $T_2$.
We call a subforest $F_1$ or $F_2$ which appears in the above dynamic programming procedure a \emph{relevant} subforest.
Klein showed the following bound on the number of relevant subforests $F_1$ of $T_1$ generated by the DP procedure.
\begin{lemma}[Lemma 3 of \cite{Klein98}]
\label{lm:heavy-light}
    If we use top-down dynamic programming to compute $\ed(T_1,T_2)$ with respect to the recursion defined in \Cref{eq:trivial,eq:right-recurse,eq:left-recurse},
    we only ever need to compute $\ed(F_1, F_2)$ for $O(|T_1|\log |T_1|)$ distinct subforests $F_1$ of $T_1$.
\end{lemma}
The proof of this lemma uses a heavy-light decomposition argument, which crucially relies on choosing the ``direction'' of recursion (\Cref{eq:right-recurse,eq:left-recurse}) based on the sizes of the leftmost and rightmost trees in $F_1$. This improves upon the previous DP algorithm by Zhang and Shasha \cite{ZhangS89}, which always recurses on the rightmost roots and could only give an $O(|T_1|^2)$ bound instead of $O(|T_1|\log |T_1|)$.  

Since there are only $O(|T_2|^2)$ possible subforests $F_2$ of $T_2$ (\Cref{prop:n2}), \Cref{lm:heavy-light} shows that we can compute $\ed(T_1, T_2)$ in $O(|T_1||T_2|^2\log |T_1|)$ time.
In the next section, we show how to use the assumption that $\ed(T_1, T_2) \le k$ to bound the number of relevant $F_2$ as well, and through this get a faster algorithm.

\section{Improved Algorithm}
\label{sec:improve}
\subsection{DP state transition graph}
\label{sec:dp-dag}
Our algorithm builds on Klein's DP algorithm described in \Cref{sec:klein}. For the sake of analysis, it is helpful to consider the DP state transition graph, which is a directed acyclic graph with vertices representing the DP states $(F_1,F_2)$ and edges representing DP transitions. 
Each edge is associated with a proxy cost that lower bounds the true incurred cost when using this transition in the actual DP. 
These will be based off the trivial lower bound
\begin{equation}
\label{eq:size-lb}
\ed(F_1, F_2) \ge \left|\size(F_1) - \size(F_2)\right|,
\end{equation}
which holds because each operation changes the size of a tree by at most 1, and at the end of applying $\ed(F_1, F_2)$ operations the trees must have the same size.

To define the DP state transition graph, we distinguish three types of DP transition that can occur from following the recursion of Klein's algorithm described in \Cref{eq:right-recurse,eq:left-recurse}. 
The first type corresponds to the first two cases of \cref{eq:right-recurse,eq:left-recurse} where we delete the rightmost or leftmost root of the forest.
The second and third types of transition capture the two subproblems generated from the third case of \cref{eq:right-recurse,eq:left-recurse} where we match nodes in the trees. Hence, the edges in the DP state transition graph and their proxy costs are defined as follows:
\begin{description}
    \item[Type 1 (Node Removal)] 
    We delete the rightmost (or leftmost) root of $F_1$ (or $F_2$).
    
    For example, we can transition
    $(F_1,F_2) \to (F_1 - r_{F_1}, F_2)$. 
    This transition has cost $1$.
   \item[Type 2 (Subtree Removal)] 
   We remove the rightmost (or leftmost) subtrees of $F_1$ and $F_2$.
   
   For example, we can transition
   $(F_1,F_2) \to (L'_{F_1}, L'_{F_2})$.
   This transition costs at least $|\size(R_{F_1}) - \size(R_{F_2})|$ by \cref{eq:size-lb} and the last case of \cref{eq:right-recurse}.
   \item[Type 3 (Subtree Selection)] 
   We focus on the subtrees below the rightmost (or leftmost) roots of $F_1$ and $F_2$.
   
   For example, we can transition
   $(F_1,F_2) \to (R^\circ_{F_1}, R^\circ_{F_2})$. 
   This transition costs at least  
   $ | \size(L'_{F_1}) - \size(L'_{F_2}) |$ by \cref{eq:size-lb} and the last case of \cref{eq:right-recurse}.
\end{description}

\subsection{Pruning DP states}
Each pair of subforests $(F_1, F_2)$ is a potential state in the DP table.
We say a state $(F_1,F_2)$ is ``not useful'' or \emph{useless} if we do not need to evaluate $\ed(F_1,F_2)$ to compute the overall tree edit distance $\ed(T_1,T_2)$.
Having defined the DP state transition graph, we use the following simple observation to label some states as useless.
\begin{proposition}[DP State Pruning Rule 1]
\label{prop:rule1}
Suppose input trees $T_1,T_2$ satisfy $\ed(T_1,T_2)\le k$.
Then if a state cannot be reached from $(T_1,T_2)$ by traversing a sequence of edges with total cost at most $k$ in the DP state transition graph, that state is useless.
\end{proposition}
We will also make use of the following pruning rule, which is a direct application of \cref{eq:size-lb}.
\begin{proposition}[DP State Pruning Rule 2]
\label{prop:rule2}
Suppose input trees $T_1,T_2$ satisfy $\ed(T_1,T_2)\le k$.
If $|\size(F_1)-\size(F_2)| > k$, then the DP state $(F_1,F_2)$ is useless.
\end{proposition}

The two pruning rules will enable us to prove the following core result, which shows that when the tree edit distance is bounded, each relevant subforest cannot occur in too many useful states.
\begin{lemma}[Number of useful DP states]
\label{lm:bound-subforests}
Suppose input trees $T_1,T_2$ satisfy $\ed(T_1,T_2)\le k$.
For each relevant subforest $F_1$ of $T_1$, there are at most $O(k^2)$ subforests $F_2$ of $T_2$ such that $(F_1,F_2)$ is a useful DP state.
\end{lemma}

\Cref{lm:bound-subforests} together with \Cref{lm:heavy-light} immediately shows an $O(nk^2 \log n)$ bound on the number of useful DP states, which will suffice to prove \Cref{thm:bounded-ted},
so in the remainder of this section, we setup the proof of this lemma.

\begin{definition}[Upper parts]
Given a subforest $F$ of $T$, we partition the nodes of $T\setminus F$ into three disjoint upper parts $MU_F, LU_F$, and $RU_F$ as follows. 
\begin{itemize}
    \item  The middle upper part $MU_F$ contains the nodes on the path from the root of $T$ to $\lca(F)$ (excluding $\lca(F)$ if $\lca(F)\in F$).
    \item  The left upper part  is defined as $LU_{F}:=\{u \in T\setminus MU_{F} \mid \pre(u)<\pre(v) \text{ for all $v\in F$}\}$, where $\pre(u)$ denote the index of $u$ in the preorder traversal of $T$ ($1\le \pre(u)\le |T|$). The right upper part $RU_{F}$ is defined symmetrically using the postorder traversal of $T$. Intuitively, $LU_{F}$ consists of the nodes to the left of the path $MU_{F}$, and  $RU_F$ consists of the nodes to the right of this path.
\end{itemize}
\end{definition}
See \cref{fig:upper-parts} for some examples.
\begin{figure}[t]
\centering
\def\svgwidth{\linewidth}
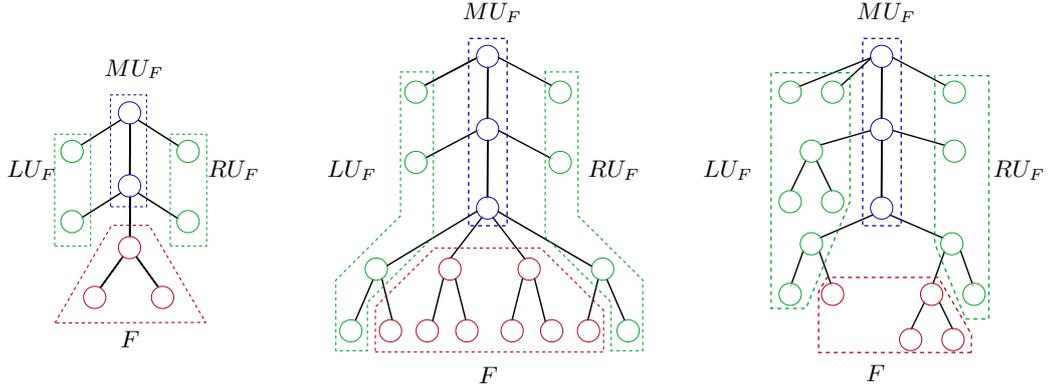
\caption{Three examples of subforests $F$ in different underlying trees, with upper parts labeled.}
\label{fig:upper-parts}
\end{figure}


If a DP state $(G_1, G_2)$ can be reached from $(T_1, T_2)$ in the DP state transition graph, it means that we obtain $G_1$ and $G_2$ by removing some nodes in $T_1$ and $T_2$ respectively, following the DP transition rules.
We classify the removed nodes in $T_1\setminus G_1$ according to which of the three upper parts they belong to.
For node $v \in T_1\setminus G_1$, if $v\in LU_{G_1}$ (or $v\in RU_{G_1}$, $v\in MU_{G_1}$), then we say $v$ is \emph{left-removed} (or \emph{right-removed}, \emph{middle-removed}) with respect to subforest $G_1$. If during a DP transition $(F_1,F_2)\to (G_1,G_2)$, a node $v\in F_1\setminus G_1$ is left-removed (or right-removed, middle-removed) with respect to not only $G_1$, but also all subforests $G_1'\subseteq G_1$ (which may be reached by later DP transitions), then we simply say $v$ is left-removed (or right-removed, middle-removed) during this DP transition, without specifying the subforest $G_1$.
The above discussion also similarly applies to the second input tree $T_2$ and its subforests.

By inspecting the DP transition rules described in \Cref{sec:dp-dag}, we immediately have the following simple but useful observation.
\begin{lemma}
\label{lm:transition-contract}
Let $(F_1, F_2)$ be a DP state. The following hold:
\begin{itemize}
    \item 
A type 2 transition from this state either right-removes $\size(R_{F_1})$ nodes, or left-removes $\size(L_{F_1})$ nodes from $F_1$, depending on whether the right or left subtree were removed.
    \item 
 A type 3 transition from this state either left-removes $\size(F_1-R_{F_1})$ nodes and middle-removes one node, or right-removes $\size(F_1 - L_{F_1})$ nodes and middle-removes one node from $F_1$, depending on whether the transition zoomed in on the right or left subtree.
\end{itemize}
Similar statements hold for removals in $F_2$.
\end{lemma}
Note that in the case of type 1 transitions, we cannot tell whether the node being removed was a left, middle, or right-removal. 
However, we observe that a type 1 transition always has cost 1. 
Combining this observation with \Cref{lm:transition-contract} and the pruning rule in \Cref{prop:rule1}, we obtain the following property of useful DP states $(G_1,G_2)$:
\begin{lemma}
\label{lm:inequality LU RU}
If DP state $(G_1,G_2)$ survives the pruning rule in \Cref{prop:rule1}, then 
 \[|\size(LU_{G_1})-\size(LU_{G_2})| \le k,\] and \[|\size(RU_{G_1})-\size(RU_{G_2})| \le k.\]
\end{lemma}
\begin{proof}
Consider the sets $LU_{G_1}$ and  $LU_{G_2}$ of left-removed nodes in $G_1$ and $G_2$.
Suppose $k_1$ nodes of $LU_{G_1}$ and $k_2$ nodes of $LU_{G_2}$ were removed by type 1 transitions,
incurring a total cost of $k_1+k_2$. 
The remaining $\size(LU_{G_1}) - k_1$ nodes in $LU_{G_1}$ and $\size(LU_{G_2}) - k_2$ nodes in $LU_{G_2}$ must be the result of type 2 and 3 transitions.

From \cref{lm:transition-contract} and the discussion in \Cref{sec:dp-dag}, we know that when a type 2 or 3 transition $t$ left-removes $c_1^{(t)}$ nodes from $T_1$ and $c_2^{(t)}$ nodes from $T_2$, the incurred cost is at least $|c_1^{(t)}-c_2^{(t)}|$. 
Then by triangle inequality, the total cost from all type 2 and 3 transitions is at least
    \[\sum_{t} \left|c_1^{(t)} - c_2^{(t)}\right| \ge \left|\sum_{t} c_1^{(t)} - \sum_{t}c_2^{(t)}\right| = 
    \left|\grp{\size(LU_{G_1}) - k_1} - \grp{\size(LU_{G_2}) - k_2}\right|,\]
where the sum is over all type 2 and 3 transitions $t$ leading from state $(T_1,T_2)$ to state $(G_1,G_2)$.
Then, by applying triangle inequality once more, the total cost from all transitions is at least
\[k_1 + k_2 + \left|\grp{\size(LU_{G_1}) - k_1} - \grp{\size(LU_{G_2}) - k_2}\right| \ge \left|\size(LU_{G_1}) - \size(LU_{G_2})\right|.\]

This proves the first inequality.
The second inequality follows from identical reasoning, applied to the right-removed instead of the left-removed nodes of $G_1$ and $G_2$.
\end{proof}

We have just derived the useful \Cref{lm:inequality LU RU} from the first pruning rule in \Cref{prop:rule1}. To prove \Cref{lm:bound-subforests},  we still need to apply the second pruning rule in \Cref{prop:rule2} as well. 
We will use the following lemma.
\begin{lemma}
\label{lm:abc}
Given three integers $a,b,c$, the number of subforests $F$ of a tree $T$ which simultaneously satisfy $|\size(LU_{F})-a|\le k$, $|\size(RU_{F})-b|\le k$, and $|\size(F)-c|\le k$ is at most $O(k^2)$.
\end{lemma}

Before proving \Cref{lm:abc}, we show that it implies the desired upper bound on the number of useful DP states that survive both pruning rules in \Cref{prop:rule1} and \Cref{prop:rule2}.
\begin{proof}[Proof of \Cref{lm:bound-subforests} given {\Cref{lm:abc}}]

We are given a relevant subforest $F_1$ of $T_1$, and want to bound the number of subforests $F_2$ of $T_2$ such that $(F_1, F_2)$ is a useful state.
By \Cref{lm:inequality LU RU}, the state $(F_1, F_2)$ is useful only if
\[\left|\size(LU_{F_2}) - a\right|, \left|\size(RU_{F_2}) - b\right| \le k\]
for $a = \size(LU_{F_1})$ and $b = \size(RU_{F_1})$.
Moreover, by \Cref{prop:rule2} if the state is useful then
\[\left|\size(F_2) - c\right| \le k\]
for $c = \size(F_1)$.
Hence, applying \Cref{lm:abc} with $T = T_2$ immediately implies that there are $O(k^2)$ possibilities for $F_2$, which proves the desired result.
\end{proof}

\subsection{Proof of \Cref{lm:abc}}
Suppose $\size(LU_F)=\ell$ and $\size(RU_F)=r$ for some integers $\ell$ and $r$ within $k$ of $a$ and $b$ respectively.
Then we claim the following algorithm outputs all possible subforests $F$ satisfying the hypotheses of the lemma:

\begin{enumerate}
    \item  Initialize $F=T$ as the given tree.
    \item While $\ell\neq 0$ or $r\neq 0$:
        \begin{enumerate}
            \item If $F$ has only root remaining, delete this root (middle-removal) from $F$
            \item Otherwise, $F$ has more than one root remaining:
            \begin{enumerate}
                \item If $\size(L_F)\le \ell$: remove the leftmost tree and update $\ell\gets \ell- \size(L_F)$
                \item Else if $\size(R_F)\le r$: remove the rightmost tree and update $r\gets r- \size(R_F)$
                \item Otherwise remove the leftmost root $\ell$ times, remove the rightmost root $r$ times, and return $F$ \emph{(unique solution case)}\label{line:unique1}
            \end{enumerate}
        \end{enumerate}
        \item If $F$ has one root remaining: repeatedly remove the only root (middle-removal) until we no longer have a single root. Return all the forests encountered during this procedure as possible solutions of $F$ \emph{(multiple solutions case)} \label{line:multi}
        \item Otherwise, $F$ has more than one root remaining: return $F$ \emph{(unique solution case)} \label{line:unique2}
\end{enumerate}

\begin{figure}[h]
\centering
\def\svgwidth{0.6\linewidth}
\begingroup%
  \makeatletter%
  \providecommand\color[2][]{%
    \errmessage{(Inkscape) Color is used for the text in Inkscape, but the package 'color.sty' is not loaded}%
    \renewcommand\color[2][]{}%
  }%
  \providecommand\transparent[1]{%
    \errmessage{(Inkscape) Transparency is used (non-zero) for the text in Inkscape, but the package 'transparent.sty' is not loaded}%
    \renewcommand\transparent[1]{}%
  }%
  \providecommand\rotatebox[2]{#2}%
  \newcommand*\fsize{\dimexpr\f@size pt\relax}%
  \newcommand*\lineheight[1]{\fontsize{\fsize}{#1\fsize}\selectfont}%
  \ifx\svgwidth\undefined%
    \setlength{\unitlength}{198.42519685bp}%
    \ifx\svgscale\undefined%
      \relax%
    \else%
      \setlength{\unitlength}{\unitlength * \real{\svgscale}}%
    \fi%
  \else%
    \setlength{\unitlength}{\svgwidth}%
  \fi%
  \global\let\svgwidth\undefined%
  \global\let\svgscale\undefined%
  \makeatother%
  \begin{picture}(1,0.78571429)%
    \lineheight{1}%
    \setlength\tabcolsep{0pt}%
    \put(0.76171406,0.56407268){\color[rgb]{0,0,0}\makebox(0,0)[lt]{\lineheight{1.25}\smash{\begin{tabular}[t]{l}$r=2$\end{tabular}}}}%
    \put(0,0){\includegraphics[width=\unitlength,page=1]{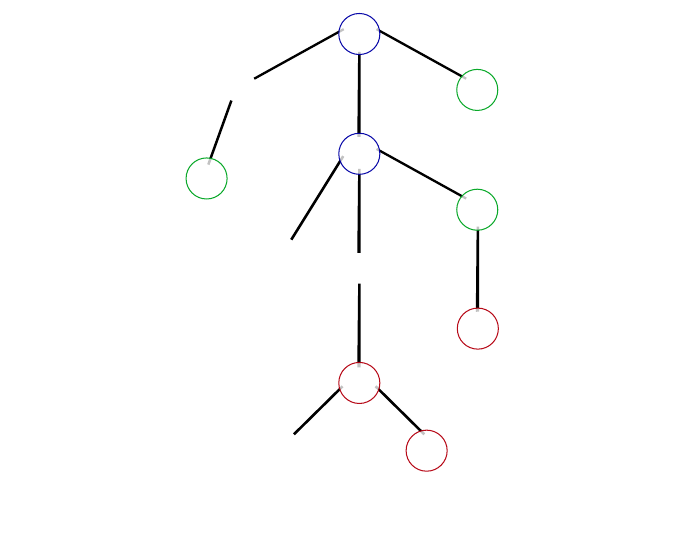}}%
    \put(0.51006235,0.02568654){\color[rgb]{0,0,0}\makebox(0,0)[lt]{\lineheight{1.25}\smash{\begin{tabular}[t]{l}$F$\end{tabular}}}}%
    \put(0.14939221,0.56407268){\color[rgb]{0,0,0}\makebox(0,0)[lt]{\lineheight{1.25}\smash{\begin{tabular}[t]{l}$\ell=2$\end{tabular}}}}%
    \put(0,0){\includegraphics[width=\unitlength,page=2]{ted-unique-sol.pdf}}%
  \end{picture}%
\endgroup%

\caption{An example of the unique solution case, where $\ell=2$ and $r=2$.}
\label{fig:example-unique}
\end{figure}
\begin{figure}[h]
\centering
\def\svgwidth{0.6\linewidth}
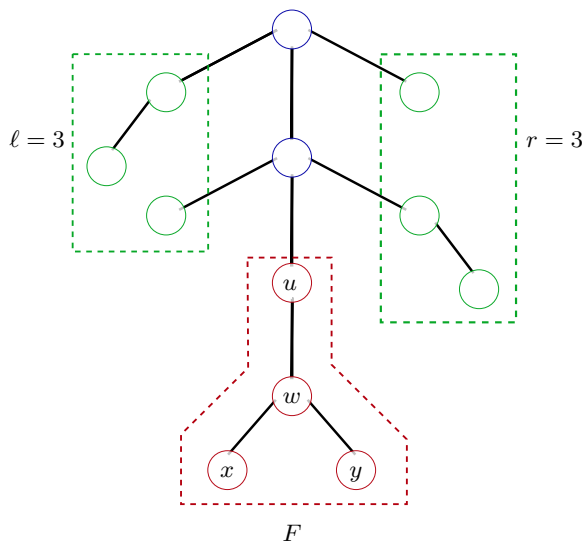
\caption{An example of the multiple solutions case, where $l=3,r=3$. Before executing step \ref{line:multi} of the algorithm, the remaining $F$ consists of $\{u,w,x,y\}$. Then the algorithm returns three possible solutions: $\{u,w,x,y\}, \{w,x,y\}$, and $\{x,y\}$.}
\label{fig:example-multi}
\end{figure}
At each step of the algorithm, we are either at a state with multiple roots or at a state with one root.
In the former case, we have to left-remove or right-remove which we do (unless we have already left-removed $\ell$ times and right-removed $r$ times, in which case we halt). 
In the latter case, if we still have not left-removed or right-removed the full number of times, we must keep middle-removing until we can make left or right removals.
Terminating in one of these states corresponds to the \emph{unique solution} cases of the algorithm (at step \ref{line:unique1} or step \ref{line:unique2}). An example is given in \Cref{fig:example-unique}.

In these situations, the algorithm halts on the unique subforest $F$ of $T$ with $\size(LU_F) = \ell$ and $\size(RU_F) = r$.
Since there are $O(k)$ possible values for $\ell$ and $r$ individually, we get that there are at most $O(k^2)$ distinct subforests $F$ which can be outputted as a ``unique solution'' in the above procedure.

The only other possibility is that we find ourselves in step \ref{line:multi} of the algorithm at a point where we have already left-removed $\ell$ times and right-removed $r$ times, and there is only one root $u$ remaining.
In this case $F$ might not be uniquely determined: we can continue to middle-remove the remaining root for some number of times and then return a possible solution of $F$. 
Formally, let $w$ be the deepest descendant of the remaining root $u$, such that for every node $v$ on the path from $w$ to $u$, $v$ has no siblings. Then, for every such node $v$, the subtree rooted at $v$ (denoted $T_v$) and $T_v-v$ can be a valid solution for $F$. This describes the \emph{multiple solutions case} annotated in step 3 of the above procedure. An example of the multiple solutions case is given in \Cref{fig:example-multi}.

By the above discussion, a subforest $F$ from the multiple solutions case can be 
determined uniquely by the identity of the lowest common ancestor $v = \lca(F)$, and the choice of whether $v$ is in $F$ or not.
We now prove that, over all choices of valid $\ell$ and $r$, there are only $O(k)$ many possibilities for the node $v$. Combined with the unique solution case, this will immediately finish the proof of the lemma.

We first consider the case where, among all possible node choices for $v$, there are two such that neither is an ancestor of the other. 
Then pick the leftmost (with respect to post-order traversal) and the rightmost (with respect to preorder traversal) of such possibilities for $v$, denoted $v_1$ and $v_2$ respectively.
Let $G_1$ and $G_2$ be the subtrees in $T$ rooted at $v_1$ and $v_2$, respectively.
Then by the assumptions on $F$ we necessarily have \[|\size(RU_{G_1})-b|, |\size(RU_{G_2})-b|\le k.\] 
Note that $RU_{G_2}\subseteq RU_{G_1}$.
Write $D = RU_{G_1}\setminus RU_{G_2}$ for the difference of the right upper part of $G_1$ and the right upper part  $G_2$.
Thus, by triangle inequality, we get that
\[\size(D) = \size(RU_{G_1})- \size(RU_{G_2})   \le 2k.\]

By our choice of $v_1$ and $v_2$, we know that any possible choice for $v$ is either a node in  $D$ or an  ancestor of $v_1$. 
For the former case, we have already shown that there are at most $O(k)$ nodes in $D$.
In the latter case, each distinct $v$ which is an ancestor of $v_1$ determines a subforest $F$ of a different size.
Then because we are assuming that $|\size(F) - c| \le k$, there are only $O(k)$ possibilities for the choices of $v$ which are ancestors of $v_1$.

The previous argument applies whenever there are two choices for $v$, neither of which is an ancestor of the other.
If there do not exist such options for $v$, then all possible choices of $v$ lie on the a single root-to-leaf path of $T$. 
By the same reasoning as before, the number of possible cases for $v$ here is again at most $O(k)$, because each $v$ would determine a different-sized subforest and $\size(F)$ is allowed to take on $O(k)$ distinct values.

This completes the proof of \Cref{lm:abc}.
As noted earlier, this implies \Cref{lm:bound-subforests}.
We conclude by tying these results back to our main theorem.

\begin{proof}[Proof of \Cref{thm:bounded-ted}]
Set up a table which can be indexed by pairs of subforests $(F_1, F_2)$ of $T_1$ and $T_2$.
Begin using Klein's dynamic programming approach outlined in \Cref{sec:klein} and \Cref{lm:heavy-light} but avoid generating subproblems according to the pruning rules described in \Cref{prop:rule1} and \Cref{prop:rule2},
and store solutions $\ed(F_1, F_2)$ produced.
In particular, when Klein's algorithm would normally generate a subproblem, we first check if the produced subproblem would be a useful state according to our previous definitions.
\Cref{prop:rule2} and the proof of \Cref{lm:abc} make it clear that we can quickly check if a state is useful provided we know the sizes of $F_1, F_2, LU_{F_1},RU_{F_1}, LU_{F_2}$, and $RU_{F_2}$, and this information can be kept track of easily simply by updating the sizes according to the type of transition we follow in the table.

So, we can compute $\ed(T_1, T_2)$ while only computing $\ed(F_1, F_2)$ for useful states.
By \Cref{lm:heavy-light} there are $O(n\log n)$ possibilities for $F_1$ and by \Cref{lm:bound-subforests} there are $O(k^2)$ choices for $F_2$ for each $F_1$.
So overall we only fill in at most $O(nk^2\log n)$ entries of the DP table.
Since we do a constant amount of work to get the value at each entry of the table, our algorithm has the desired running time.
\end{proof}

\section{Open problems}
\label{sec:open}
For trees of bounded edit distance $k = O(1)$ our algorithm runs in linear time.
However, for larger tree edit distances
$k = \Theta(n)$ our algorithm requires $O(n^3\log n)$ time, which is slower than the fastest known algorithm \cite{demaine2009} for general tree edit distance
by a logarithmic factor.
This motivates the question: can we solve the bounded tree edit distance problem in $O(nk^2)$ time instead of $O(nk^2 \log n)$?

The easier problem of \emph{string} edit distance can be solved in $\tilde O(n+k^2)$ time \cite{Meyers86,LandauV88}, which is quasilinear even for super constant distance parameter $k = O(\sqrt n)$.
This motivates the question of whether it is possible to get similar speedups for tree edit distance. It would be especially interesting to see if the bounded tree edit distance problem can be solved in $\tilde O(n+k^3)$ time. 
Perhaps the suffix tree techniques used in \cite{ShashaZ90} (and discussed in \cite[Appendix]{apxfocs2019}) could prove useful in showing such a result.

Regarding variants of tree edit distance, it remains an open question to get faster algorithms for the harder problem of \emph{unrooted} tree edit distance \cite{Klein98, unrooted2018} (where the elementary operations are edge contraction, insertion, and relabeling) when the distance is bounded by $k$.  
The best known algorithm for unrooted tree edit distance was recently given by Dudek and Gawrychowski \cite{unrooted2018} and runs in $O(n^3)$ time. The previous $O(n^3 \log n)$ time algorithm by Klein \cite{Klein98} also applies to the unrooted setting. Although we extended Klein's algorithm to tackle the rooted tree edit distance problem in $O(nk^2 \log n)$ time, it is not obvious how to extend their approach to the unrooted bounded distance setting.
This is because Klein solves the unrooted version of the problem by dynamic programming over the subproblems generated by all possible rootings of $T_2$.
This is fine for computing general edit distance because the number of subforests over all possible rootings is $O(n^2)$ just like the number of subforests for a fixed rooted tree on $n$ nodes.
However, when the tree edit distance is bounded, the number of possible relevant subproblems over all possible rootings can be $\Omega(n)$ even when $k$ is small.
Although our algorithm can be used to recover a near quadratic time algorithm for unrooted tree edit distance when $k = O(1)$ is constant, it remains open whether we can obtain a quasilinear time algorithm in this setting.

Finally, although general tree edit distance with \emph{arbitrary weights} cannot be solved in truly subcubic time unless certain popular conjectures are false \cite{apsphard2020}, analogous fine-grained hardness results rule out truly subquadratic time algorithms for string edit distance even when deletions and insertions have unit cost \cite{BackursI18}.
Can we show conditional hardness for tree edit distance with unit costs, or can we find a subcubic time algorithm for this problem? 



\bibliography{theory}

\end{document}